\documentclass[11pt]{article}

\usepackage[utf8]{inputenc}	
\usepackage{amsmath,amsthm,amsfonts,amssymb,amscd}
\usepackage{tikz}
\usetikzlibrary{shapes.geometric, positioning, math}
\usepackage{multirow,booktabs}
\usepackage[table]{xcolor}
\usepackage{fullpage}
\usepackage{lastpage}
\usepackage{enumitem}
\usepackage{fancyhdr}
\usepackage{mathrsfs}
\usepackage{wrapfig}
\usepackage{setspace}
\usepackage{calc}
\usepackage{multicol}
\usepackage{cancel}
\usepackage[margin=1in]{geometry}
\usepackage{amsmath}
\usepackage{diagbox}

\usepackage{empheq}
\usepackage{mdframed}
\usepackage[most]{tcolorbox}
\usepackage{xcolor}
\usepackage{algorithm}
\usepackage{algorithm}
\usepackage{algpseudocode}
\colorlet{shadecolor}{orange!15}
\newtheorem{theorem}{Theorem}[section]

\newtheorem{lemma}[theorem]{Lemma}
\newtheorem{proposition}[theorem]{Proposition}
\theoremstyle{definition}
\newtheorem{definition}[theorem]{Definition}

\usepackage{hyperref}

\newcommand{\Z}{\mathbb{Z}}
\newcommand{\R}{\mathbb{R}}

\title{Multi-tier Flexible Graph Connectivity\thanks{Grainger College of Engineering, Univ. of Illinois, Urbana-Champaign, Urbana, IL 61801. Email: {\tt\{karthe, ryjiang2, kk17\}@illinois.edu}. Supported in part by NSF grant CCF-2402667.}}
\author{Karthekeyan Chandrasekaran \and Raymond Jiang \and Krishna Kalathur}

\begin{document}
\date{}
\maketitle

\begin{abstract}
    Motivated by non-uniform edge failures in network design, we introduce a multi-tier model of flexible graph connectivity. In $k$-tier Flexible Graph Connectivity ($k$-tier FGC), the input is an undirected graph $G=(V, E)$ with non-negative edge costs, along with a classification of the edges into nested tiers $T_1\subseteq T_2\subseteq \ldots\subseteq T_k=E$ and non-negative integral tier requirements $q_1\le q_2\le \ldots\le q_k$. A cut $\emptyset \neq R\subsetneq V$ is safe if it is safe along one of the tiers, i.e., there exists $i\in [k]$ such that $|\delta(R)\cap T_i|\ge q_i$. The goal is to find a minimum-cost subset $F\subseteq E$ of edges such that the subgraph $(V, F)$ has no unsafe cuts. The case of $k=1$ corresponds to the min-cost $q_1$-edge-connected spanning subgraph problem, which is APX-hard. We design approximation algorithms for every fixed constant $k$ for three variants of $k$-tier FGC: (i) for $k$-tier FGC, we design an LP-based logarithmic approximation, (ii) for min-cardinality $k$-tier FGC, we design a combinatorial approximation whose factor depends only on the tier requirements $q_1$ and $q_k$, and (iii) for $k$-tier Flexible Multi-Graph Connectivity, where we are allowed to use multiple copies of each edge while paying the cost of the edge for each chosen copy of the edge, we design an LP-based $2$-approximation. 

\end{abstract}

\section{Introduction}
Network design is a fundamental area of combinatorial optimization with applications in communication, transportation, and infrastructure planning. The typical network design problem asks for a minimum-cost subgraph that satisfies certain connectivity requirements. A long line of work has produced powerful techniques---e.g., augmenting paths, LP duality, and matroid theory towards fast algorithms for polynomial-time problems such as shortest path, min-cost flow, and min-cost spanning tree, as well as combinatorial augmentation, primal-dual methods, and iterative rounding towards approximation algorithms for NP-hard problems such as min-cost Steiner tree and min-cost $k$-edge-connected spanning subgraph (denoted min-cost $k$-ECSS) \cite{KhullerVishkin1994, KR96, Jain2001, GGTW09}. Classical models, however, treat all edges uniformly: every edge in the network is equally likely to be attacked and hence, fault-tolerance requirements treat all edges to be at the same failure-level. Real-world networks rarely satisfy the uniform failure model. Some links are protected by redundancy while others are exposed. This motivates the study of non-uniform failure models, in which different sets of edges are subject to failure at different rates. 

Adjiashvili, Hommelsheim, and M\"{u}hlenthaler \cite{AHM2022} introduced $(p, q)$-Flexible Graph Connectivity problem (denoted $(p, q)$-FGC) to model non-uniform failure models: 
the input consists of a graph $G=(V, E)$ whose edges are partitioned into ``safe'' and ``unsafe'' edges, and the goal is to buy a minimum-cost subset $F\subseteq E$ of edges such that for every choice of at most $q$ unsafe edges $F'\subseteq F$, the subgraph $(V, F\setminus F')$ is $p$-edge-connected. While $(p, q)$-FGC neatly captures the dichotomy between safe and unsafe edges, many applications involve a richer hierarchy: edges may be vulnerable in graded tiers, and the operator may demand different levels of connectivity against failure of edges in different tiers. For instance, the operator may require high connectivity when only the least robust edges can fail, but can tolerate weaker connectivity when the adversary is allowed to attack more robust edges. Such multi-tiered failure scenarios necessitate non-uniform failure models that go beyond the safe/unsafe dichotomy. 

Motivated towards addressing multi-tiered failure scenarios, we consider the \emph{$k$-tier Flexible Graph Connectivity problem}. 
\begin{mdframed}
    \noindent\textbf{$k$-tier Flexible Graph Connectivity ($k$-tier FGC).}
    
    \noindent\textbf{Given.} An undirected graph $G=(V, E)$ with non-negative edge costs $c: E\to\R_{\ge 0}$, a nested family of edge-tiers 
    $\mathcal{T}=(T_1, T_2, \ldots, T_k)$ where 
    $T_1\subseteq T_2\subseteq \ldots \subseteq T_k=E$, and non-negative integral tier requirement vector $\mathcal{Q}=(q_1, q_2, \ldots, q_k)\in \mathbb{Z}_{\geq 0}^k$ where $q_1\le q_2\le \ldots \le q_k$. 
    
    \noindent\textbf{Definition.} A cut $\emptyset \neq R\subsetneq V$ is \emph{$(\mathcal{T}, \mathcal{Q})$-safe} if there exists $i\in [k]$ such that $|\delta(R)\cap T_i|\ge q_i$ and \emph{unsafe} otherwise. The graph $G$ is $(\mathcal{T}, \mathcal{Q})$-FGC if every cut $\emptyset\neq R\subsetneq V$ is $(\mathcal{T}, \mathcal{Q})$-safe in $G$. 

    \noindent\textbf{Goal.} A minimum-cost subset $F\subseteq E$ such that the subgraph $(V, F)$ is $(\mathcal{T}, \mathcal{Q})$-FGC.  
\end{mdframed}
We observe that verifying whether a given candidate $F\subseteq E$ is feasible for $k$-tier FGC is itself non-trivial: one must check that for each cut $R$ whether at least one of the $k$ tier inequalities hold. We show that this reduces to a multi-objective minimum-cut problem (see Lemma \ref{lem:verification-to-multiobj-min-cut}) and is hence, polynomial-time solvable for every fixed constant $k$. Multi-objective minimum-cut is NP-hard if the number of objectives is part of input \cite{AZ06}. For this reason, we restrict our focus to fixed constants $k$ throughout this work since feasibility is efficiently verifiable only in this regime. We also assume throughout that the input instance is feasible.  

Multi-tier FGC subsumes several classical and recent network-design problems. The case of $k=1$ corresponds to min-cost $q_1$-ECSS: we observe that a cut is safe if and only if it has at least $q_1$ edges in $T_1=E$. The case of $k=2$ corresponds to $(q_1, q_2-q_1)$-FGC: here $T_1$ plays the role of safe edges and $T_2-T_1$ plays the role of unsafe edges. Multi-tier FGC therefore unifies and substantially generalizes both classical edge-connectivity and the recent flexible-connectivity problems. 
While the unification and generalization make it a compelling model, it is to be noted that many of the known techniques for $1$ and $2$ tier-FGC for small values of tier-requirements---uncrossability arguments and cover small cuts---do not appear to extend to a larger number of tiers (e.g., $k=3$).

The case of $k=1$ corresponds to min-cost $q_1$-ECSS which is APX-hard already for $q_1=2$ \cite{Fer98}, while Jain's seminal $2$-approximation for survivable network design \cite{Jain2001} applies to min-cost $q_1$-ECSS for all $q_1$. The case of $k=2$ corresponds to $(p, q)$-FGC which has been intensively studied since its introduction in 2020. A sequence of works has improved the approximation factor for $(p, q)$-FGC for various parameter regimes \cite{AHM2022, BCGI23, BHKS24, CJ25, BCKS25, Nut25, IV25, Ban25, Nut25-tight}. Most relevant to our work, Ibrahimpur and V\'{e}gh \cite{IV25} recently obtained an $O(\log{n})$-approximation for $(p, q)$-FGC for all $p$ and $q$ via an LP-relaxation and an independent rounding analysis, where $n$ is the number of vertices in the input graph. We generalize their result to obtain a log-approximation for $k$-tier FGC. 

\begin{theorem}\label{thm:log-approx}
    There exists a polynomial-time randomized algorithm for $k$-tier FGC that returns, with probability at least $1/3$, a feasible solution with approximation factor $O(k^2 \log{n})$. 
\end{theorem}
We recall that $k$ is a fixed constant and hence, the approximation factor in Theorem \ref{thm:log-approx} is $O(\log{n})$. 
Our algorithm for Theorem \ref{thm:log-approx} is LP-based. Our LP is a \emph{weakening} of the Ibrahimpur-V\'{e}gh LP for $(p, q)$-FGC. We weaken it since it enables a simpler and easy-to-generalize analysis without losing on the logarithmic approximation factor. We also exhibit an instance with an integrality gap of $2k$ for our LP. 

We next consider the \emph{minimum-cardinality} version of the problem, where all edge costs are unit. For $p$-ECSS, the cardinality version has much better approximation factor than the cost version: \cite{CT00, GGTW09} give $(1+O(1/p))$-approximation. We let $\alpha_p$ denote the best-known approximation ratio for min-cardinality $p$-ECSS. For min-cardinality $(p,q)$-FGC, Nutov \cite{Nut25} gave a combinatorial $(\alpha_p+2q/p)$-approximation. We generalize Nutov's ideas to the multi-tier setting:
\begin{theorem}\label{thm:cardinality}
There is a polynomial-time $(\alpha_{q_1}+ 2(q_k-q_1)/q_1)$-approximation for min-cardinality $k$-tier FGC, where $\alpha_{q_1}$ is the best known approximation ratio for min-cardinality $p$-ECSS. 
\end{theorem}

Next, we consider the multi-use model of the problem termed $k$-tier Flexible Multigraph Connectivity problem (denoted $k$-tier FMGC). Here, for the same inputs, the problem seeks a non-negative integral vector $x\in \Z^E_{\ge 0}$ minimizing $\sum_{e\in E}c_e x_e$ such that for every $\emptyset\neq R\subsetneq V$, there exists $i\in [k]$ with $x(\delta(R)\cap T_i)\ge q_i$. This relaxes the constraint on each edge, allowing multiplicity. For the case of $k=1$, namely min-cost $p$-edge connected spanning multigraph (denoted $p$-ESCM), a  $(1+O(1/p))$-approximation is known and this is the best possible approximation factor \cite{HKZ25, KS25}. For the case of $k=2$, namely $(p, q)$-Flexible Multigraph Connectivity (denoted $(p, q)$-FMGC), Simmons' \cite{simmons-thesis} designed a $2$-approximation via the cut-based formulation for min-cost spanning tree. We extend Simmons' approach to $k$-tier FMGC. 
\begin{theorem}\label{thm:multi-use}
    There is a polynomial-time $2$-approximation for $k$-tier FMGC. 
\end{theorem}

We note that after the conference publication of \cite{IV25}, Ibrahimpur and V\'{e}gh independently also introduced the multi-tier FGC model and proved Theorem \ref{thm:log-approx} in their journal publication \cite{Ibrahimpur2026journal}. Our model and results were discovered simultaneously and independently after the conference publication and before the journal publication of Ibrahimpur and V\'{e}gh's work.

\subsection{Feasibility Verification}
In this section, we reduce the problem of verifying whether all cuts are safe to a multiobjective min-cut problem.  We begin with a definition of the multiobjective min-cut problem. 
\begin{mdframed}
    \noindent\textbf{Multiobjective Global Min Cut.}
    
    \noindent\textbf{Given.} An undirected graph $G=(V, E)$ with non-negative edge costs $c_1, c_2, \ldots, c_k: E\to\R_{\ge 0}$. 
    
    \noindent\textbf{Goal.} $\min_{\emptyset\neq R\subseteq V}\max_{i\in [k]}c_i(\delta(R))$. 
\end{mdframed}
We need the following result on multiobjective global min-cut. 
\begin{lemma}[\cite{AZ06}] \label{lem:multiobj-min-cut}
    There exists an algorithm that runs in time $O(mn^{2k})$ to solve multiobjective global min-cut, where $m$ is the number of edges and $n$ is the number of vertices in the input graph. Moreover, if the optimum value $\lambda$ of the multiobjective global min-cut is strictly positive, then 
    $|\{\emptyset\neq R\subsetneq V: \max_{i\in [k]}c_i(\delta(R))=\lambda\}|=O(n^{2k})$, and all non-empty proper subsets $R\subset V$ with  $\max_{i\in [k]}c_i(\delta(R))=\lambda$ can be enumerated in time $O(mn^{2k})$. 
\end{lemma}

\begin{lemma}\label{lem:verification-to-multiobj-min-cut}
    There exists an algorithm that takes as input an undirected graph $G = (V,E)$ with edge-tiers $\mathcal{T}=(T_1, T_2, \ldots, T_k)$ where $T_1\subseteq T_2 \subseteq \ldots \subseteq T_k = E$ and non-negative integral edge-tier requirements $\mathcal{Q}=(q_1, q_2, \ldots, q_k)$ where $q_1\le q_2\le \ldots\le q_k$,  
    and runs in time $|E||V|^{O(k)}$ to verify whether $G$ is $(\mathcal{T}, \mathcal{Q})$-FGC. 
\end{lemma}
\begin{proof}
If $q_1=0$, then $G$ has no $(\mathcal{T}, \mathcal{Q})$-unsafe cuts. Henceforth, we assume that $q_1>0$. 
We observe that $G=(V, E)$ has no $(\mathcal{T}, \mathcal{Q})$-unsafe cuts if and only if $\underset{i \in [k]}{\max}\left\{\frac{|\delta(R)\cap T_i|}{q_i}\right\} \geq 1$ for each subset $\emptyset \neq R \subsetneq V$. Thus, $G$ has no $(\mathcal{T}, \mathcal{Q})$-unsafe cuts iff $\underset{\emptyset \neq R \subsetneq V}{\min}\underset{i \in [k]}{\max}\left\{\frac{|\delta(R)\cap T_i|}{q_i}\right\} \geq 1$. 
For each $i \in [k]$, we define $c_{i}(e): E \rightarrow \mathbb{R}_{\geq 0}$ as
\begin{equation*}
    c_i(e) :=
    \begin{cases}
        \frac{1}{q_i} & \text{ if } e\in T_i,\\
        0 &\text{otherwise}.
    \end{cases}
\end{equation*}
 Then, verifying whether $\underset{\emptyset \neq R \subsetneq V}{\min}\underset{i \in [k]}{\max}\left\{\frac{|\delta(R)\cap T_i|}{q_i}\right\} \geq 1$ holds can be done by solving the 
 multi-objective global min-cut problem on graph $G=(V, E)$ with edge costs $c_1, c_2, \ldots, c_k: E\rightarrow \R_{\ge 0}$. Hence, the result follows by Lemma \ref{lem:multiobj-min-cut}.

\end{proof}

\section{Logarithmic Approximation}
In this section, we prove Theorem \ref{thm:log-approx}. For this, we formulate an IP for $k$-tier FGC, show that the LP-relaxation is solvable in polynomial time, and design a randomized rounding algorithm. For an input instance $(G=(V,E),c:E\rightarrow \mathbb{R}_{\geq 0},\mathcal{T}=(T_1,T_2,\ldots, T_k),\mathcal{Q}=(q_1,q_2,\ldots, q_k))$, let $S_i\coloneqq T_i\backslash T_{i-1}$ for every $i\in \{1,2,\dots,k\}$ where $T_0:=\emptyset$. We observe that $S_1,\ldots,S_k$ is a partition of $E$. For each $i\in [k]$ and $\emptyset\neq R\subsetneq V$, we define $\delta_{T_i}(R)\coloneqq \delta(R)\cap T_i$ and $\delta_{S_i}(R)\coloneqq \delta(R)\cap S_i$. In addition, we define for every $a\in \mathbb{R}$, $a^+\coloneqq \max\{a,0\}$. 

We consider the following IP:
\begin{align}\tag{$k$-tier-IP}\label{eq:IP1}
\min & \sum_{e\in E}c_e x_e \notag\\
\sum_{i=1}^k\left(\prod_{j\in [k]\backslash \{i\}} (q_j-|J\cap T_j|)^+ \right)\cdot x(\delta_{T_i}(R)-J)&\geq \prod_{i=1}^k(q_i-|J\cap T_i|)^+ \stepcounter{equation} \tag{\theequation}\label{eq:genconstraint}\\
\forall \text{ pairs } &(J, R) \text{ where }J\subseteq \delta(R) \text{ and } \emptyset\neq R\subsetneq V \notag \\
0 \le x_e&\le 1\ \forall\ e\in E \notag\\
x_e&\in \mathbb{Z}\ \forall\ e\in E. \notag
\end{align}

We show that \ref{eq:IP1} formulates $k$-tier FGC in Lemma \ref{lem:IP-correctness} below. 

\begin{lemma}\label{lem:IP-correctness}
    \ref{eq:IP1} formulates $k$-Tier FGC.
\end{lemma}
\begin{proof}
We first show that the indicator vector of every feasible solution $F\subseteq E$ satisfies constraint \eqref{eq:genconstraint} for $(J, R)$ for every $\emptyset\neq R\subsetneq V$ and $J\subseteq \delta(R)$. Let $F\subseteq E$ such that $(V, F)$ is $(\mathcal{T}, \mathcal{Q})$-FGC. Let $\emptyset\neq R\subsetneq V$ and $J\subseteq \delta(R)$. If $q_i-|J\cap T_i|\le 0$ for some $i\in [k]$, then RHS of \eqref{eq:genconstraint} is 0 while LHS of \eqref{eq:genconstraint} is non-negative owing to the non-negativity constraints on $x$ and therefore, the constraint holds. Thus, we may assume that $q_i-|J\cap T_i|> 0$ for every $i\in [k]$. Since $(V,F)$ is $(\mathcal{T}, \mathcal{Q})$-FGC, there exists $t\in [k]$ such that $|\delta_{T_t}(R)\cap F|\geq q_{t}$. Consequently, we have that $|(\delta_{T_t}(R)\cap F)-J|\geq q_{t}-|J\cap T_t|$. Hence, $\frac{|(\delta_{T_t}(R)\cap F)-J|}{q_{t}-|J\cap T_j|}\ge 1$ and therefore, 
  \begin{equation}
      \sum_{i=1}^k \frac{|(\delta_{T_i}(R)\cap F)- J)|}{q_i-|J\cap T_i|} \ge \frac{|(\delta_{T_{t}}(R)\cap F)-J|}{q_{t}-|J\cap T_t|}\geq 1 \label{eq:exactProofInequality}.
  \end{equation}
  Multiplying both sides of \eqref{eq:exactProofInequality} by $\prod_{i=1}^k(q_i-|J\cap T_i|)$, we obtain 
  \begin{equation*}
      \sum_{i=1}^k \left(\prod_{j\in [k]\setminus\{i\}}(q_j-|J\cap T_j|)\right)\cdot |(\delta_{T_i}(R)\cap F)- J)| \ge \prod_{j\in [k]}(q_j-|J\cap T_j|),
  \end{equation*}
  and hence, the indicator vector of $F$ satisfies constraint \eqref{eq:genconstraint}. 

  Next we show that if $F\subseteq E$ is such that $(V, F)$ is not $(\mathcal{T}, \mathcal{Q})$-FGC, then the indicator vector $x$ of $F$ violates constraint \eqref{eq:genconstraint} for some tuple $(J, R)$ where $\emptyset\neq R\subsetneq V$ and $J\subseteq \delta(R)$. Since $(V, F)$ is not $(\mathcal{T}, \mathcal{Q})$-FGC, there exists $\emptyset\neq R\subsetneq V$ such that $|\delta_{T_i}(R)\cap F|<q_i$ for all $i\in [k]$. Let $J:=\delta(R)\cap F$. For this choice of $x$ and $(J,R)$, the LHS of \eqref{eq:genconstraint} is 0 while the RHS of \eqref{eq:genconstraint} is positive, showing violation of constraint \eqref{eq:genconstraint}.
\end{proof}

We show that the LP-relaxation of \ref{eq:IP1} is solvable in polynomial time in Section \ref{sec:LP-solvability}---see Lemma \ref{thm:LPpolysolve}. We design a rounding algorithm and analyze its approximation factor in Section \ref{sec:rounding}---see Lemma \ref{lem:rounding}. Theorem \ref{thm:log-approx} follows by Lemmas \ref{lem:IP-correctness}, \ref{thm:LPpolysolve}, and \ref{lem:rounding}. We discuss the integrality gap of our LP-relaxation and compare it to Ibrahimpur-V\'{e}gh's LP-relaxation for $(p, q)$-FGC in Section \ref{sec:LP-strength}. We need the following result on the number of approximate min-cuts. 

\begin{theorem}[\cite{Kar93}]\label{thm:kargercuts}
     Let $G=(V,E)$ with non-negative edge capacities $c: E\rightarrow \R_{\ge 0}$ and let $\lambda:=\min\{c(\delta(R)): \emptyset\neq R \subsetneq V\}$. For every $\alpha\geq 1$, we have that 
     \[
     \left|\left\{R: \emptyset\neq R\subsetneq V, c(\delta(R))\le \alpha \lambda \right\}\right| = O(|V|^{2\alpha}). 
     \]
     Moreover, $\left\{R: \emptyset\neq R\subsetneq V, c(\delta(R))\le \alpha \lambda \right\}$ can be enumerated in polynomial time.
\end{theorem} 

For the rest of this section, we fix the input graph $G=(V, E)$, tiers $\mathcal{T}=(T_1, T_2, \ldots, T_k)$ where $T_1\subseteq T_2\subseteq \ldots\subseteq T_k=E$, non-negative integral tier requirements $\mathcal{Q}=(q_1, q_2, \ldots, q_k)$ where  $q_1\le q_2\le \ldots\le q_k$. 

\subsection{Efficient solvability of the LP}\label{sec:LP-solvability}
We assume that $q_i\le |E|$ for all $i\in [k]$ since the problem is infeasible otherwise. 
The following is the main result of this section.
\begin{lemma}\label{thm:LPpolysolve}
    For every fixed constant $k$, the LP relaxation of \ref{eq:IP1} can be solved in polynomial time.
\end{lemma}
We note that the number of constraints in the LP relaxation of \ref{eq:IP1} is exponential in  $|V|$. We will show that there exists an efficient separation oracle for the LP relaxation of \ref{eq:IP1} and use the Ellipsoid method to solve the LP. 

\begin{definition}[Capacity Function]
Let $x\in [0,1]^E$. 
    Let $Q_i:=\prod_{j\in [k]\setminus \{i\}}q_j$ for every $i\in [k]$. 
    For each $i\in [k]$ and each $e\in S_i$, let 
\begin{equation*}
    u_{x}(e)\coloneqq \left(\sum_{j=i}^kQ_j\right)\cdot x_e. 
\end{equation*}
\end{definition}

We first observe the following relationship for the capacity of a cut $\delta(R)$, where $R$ is a nonempty proper subset of vertices.
\begin{proposition}\label{prop:capacitycut}
Let $x\in [0,1]^E$. Then, for
every $\emptyset\neq R\subsetneq V$, we have that
    \begin{equation}
        u_{x}(\delta(R))=\sum_{i=1}^kQ_i\cdot x(\delta_{T_i}(R))\label{eq:capacityfunctionFGC}.
    \end{equation}
\end{proposition}
\begin{proof}
    We have that 
    \begin{align*}
        u_x(\delta(R))&=\sum_{i=1}^k\left(\sum_{j=i}^kQ_j\right)\cdot x(\delta_{S_i}(R))\\
        &=\sum_{j=1}^k\sum_{i=1}^j\left(Q_j\cdot x(\delta_{S_i}(R))\right)\quad\quad
        &\text{(changing order of summation)}\\
        &=\sum_{j=1}^kQ_j\sum_{i=1}^jx(\delta_{S_i}(R))\\
        &=\sum_{j=1}^kQ_j\cdot x(\delta_{T_j}(R)), 
    \end{align*}
    where the last equation is by definition of $S_i$ for every $i\in [k]$. 
\end{proof}

As a first step towards separating over the entire family of constraints \eqref{eq:genconstraint} over all $(J\subseteq \delta(R), R\subseteq 2^V-\{\emptyset, V\})$ tuples, we solve the separation problem over the family of constraints \eqref{eq:genconstraint} for a fixed $R\in 2^V-\{\emptyset,V\}$. The following lemma shows that for a given $\emptyset\neq R\subseteq V$ and $x\in [0,1]^E$, there exists a polynomial time algorithm to verify whether $x$ satisfies constraint \eqref{eq:genconstraint} for $(J, R)$ for every $J\subseteq \delta(R)$ and if not, then return a set $J\subseteq \delta(R)$ such that \eqref{eq:genconstraint} is violated for $(J, R)$. The run-time of the algorithm in the following lemma is polynomial for constant $k$.

\begin{lemma}\label{lem:genknap} 
There exists an algorithm that takes $x\in [0,1]^E$ and $\emptyset\neq R\subsetneq V$ as input and runs in time $O(|E|^{k+2})$ to verify if there exists $J\subseteq \delta(R)$ such that 
\begin{equation}\label{eq:knapsackmin}
    \sum_{i=1}^k\left(\prod_{j\in [k]\backslash \{i\}} (q_j-|J\cap T_j|)^+ \right) \cdot x(\delta_{T_i}(R)-J)-\prod_{i=1}^k\left(q_i-|J\cap T_i|\right)^+ <0.
\end{equation}
and if so, then return such a $J$.

\end{lemma}

\begin{proof}
Suppose $J\subseteq \delta(R)$ is such that $q_i\le |J\cap T_i|$ for some $i\in [k]$. Then the LHS of (\ref{eq:knapsackmin}) is $\left(\prod_{j\in [k]\backslash \{i\}} (q_j-|J\cap T_j|)^+ \right)x(\delta_{T_i}(R)-J) \ge 0$ by non-negativity of $x$. Thus, it suffices to solve the following optimization problem: 
    \begin{equation}\label{eq:knapsackmin-opt}
        \begin{split}
        &\arg\min_{J\subseteq \delta(R): |J\cap \delta_{T_i}(R)|<q_i\forall i\in [k]}\\
        &\quad \left\{\sum_{i=1}^k\left(\prod_{j\in [k]\backslash \{i\}} (q_j-|J\cap T_j|)^+ \right) \cdot x(\delta_{T_i}(R)-J)-\prod_{i=1}^k\left(q_i-|J\cap T_i|\right)^+\right\}.
        \end{split}
    \end{equation}

We note that for a subset $J\subseteq \delta(R)$ and $i\in [k]$, we have $|J\cap T_i|=\sum_{m=1}^i|J\cap S_m|$. 
We have that a set $J^*\subseteq \delta(R)$ is an optimum solution to the optimization problem in \eqref{eq:knapsackmin-opt} if and only if $J^*$ is a minimizer of the following problem:
\begin{equation}
\begin{split}
    &\min_{\substack{a_i\in [\min\{q_i,|\delta_{S_i}(R)|\}] \forall i\in [k]: \\ \sum_{i=1}^j a_i < q_j\ \forall j\in [k]}}\ \ 
     \min_{\substack{J\subseteq \delta(R): \\ |J\cap S_i|=a_i \forall i\in [k]}} \\
    &\quad \Bigg\{ \sum_{i=1}^k \left( \prod_{j\in [k]\backslash \{i\}} \left( q_j-\sum_{m=1}^j a_m \right) \right)\cdot x(\delta_{T_i}(R)-J)
    - \prod_{i=1}^k \left( q_i-\sum_{m=1}^j a_m \right) \Bigg\}.
\end{split}
\end{equation}
The number of possible choices for $(a_1,\dots,a_k)$ with $a_i\in [\min\{q_i,|\delta_{S_i}(R)|\}]$ for each $i\in [k]$ so that $\sum_{i=1}^j a_i< q_j$ for all $j\in [k]$ is at most $\prod_{i=1}^kq_i$.

After fixing $(a_1,\cdots,a_k$), $\prod_{i=1}^{k} \left(q_i - \sum_{m=1}^{j} a_{m}\right)$ is a constant. Thus, for each choice of $(a_1,\cdots,a_k)$, it suffices to solve the following optimization problem.

\begin{equation}\label{eq:equivmin}
    \arg \min_{\substack{J\subseteq \delta(R): \\|J\cap S_i|=a_i\\\forall i\in [k]}}\sum_{i=1}^k\left(\prod_{j\in [k]\backslash \{i\}} (q_j-\sum_{m=1}^ja_m) \right)\cdot\left(\sum_{e\in \delta_{T_i}(R)\backslash J}x_e\right).
\end{equation}
We will now show that \eqref{eq:equivmin} is solvable in time $O(|E|^2)$ for every $(a_1,\dots,a_k)$. Fix a choice of $(a_1,\dots,a_k)$ and denote $\alpha_i\coloneqq \prod_{j\in [k]\backslash \{i\}}(q_j-\sum_{m=1}^ja_m)$. Then, \eqref{eq:equivmin} is equivalent to 
\begin{align*}
    \min_{\substack{J\subseteq \delta(R): \\|J\cap S_i|=a_i\\\forall i\in [k]}}\sum_{i=1}^k\alpha_i\cdot\left(\sum_{e\in \delta_{T_i}(R)\backslash J}x_e\right)
&= \min_{\substack{J\subseteq \delta(R): \\|J\cap S_i|=a_i\\\forall i\in [k]}}\sum_{i=1}^k\alpha_i\cdot\left(\sum_{m=1}^i\sum_{e\in \delta_{S_m}(R)\backslash J}x_e\right)\\
&= \min_{\substack{J\subseteq \delta(R): \\|J\cap S_i|=a_i\\\forall i\in [k]}}\sum_{i=1}^k\sum_{m=1}^i\left(\alpha_i\cdot\sum_{e\in \delta_{S_m}(R)\backslash J}x_e\right)\\
&=\min_{\substack{J\subseteq \delta(R): \\|J\cap S_i|=a_i\\\forall i\in [k]}}\sum_{i=1}^k\left(\left(\sum_{m=i}^k\alpha_m\right)\cdot\left(\sum_{e\in \delta_{S_i}(R)\backslash J}x_e\right)\right)\\
&=\sum_{i=1}^k\left(\left(\sum_{m=i}^k\alpha_m\right)\cdot\min_{\substack{J\subseteq \delta(R): \\|J\cap S_i|=a_i}}\left(\sum_{e\in \delta_{S_i}(R)\backslash J}x_e\right)\right)\\
&=\sum_{i=1}^k\left(\left(\sum_{m=i}^k\alpha_m\right)\cdot\min_{\substack{J_i\subseteq \delta_{S_i}(R): \\|J_i|=a_i}}\left(\sum_{e\in (\delta(R)\cap S_i)\backslash J_i}x_e\right)\right).
\end{align*}
The third equation is by swapping the order of sums from the second equation. The last two equations are because $S_i\cap S_j=\emptyset$ for every distinct $i,j\in [k]$. Let $i\in [k]$. We now focus on the following problem:
\begin{equation}\label{eq:inner-min}
\arg\min_{\substack{J_i\subseteq \delta_{S_i}(R): \\|J_i|=a_i}}\left(\sum_{e\in \delta_{S_i}(R)\backslash J_i}x_e\right).
\end{equation}
We observe that the optimum is achieved by the subset $J$ that consists of the $a_i$ edges in $\delta_{S_i}(R)$ that have the largest $x_e$ value. Hence, we can sort all edges in $\delta_{S_i}(R)$ in decreasing order and find the optimum set $J_i$. Therefore, \eqref{eq:inner-min} can be solved in $O(|\delta_{S_i}(R)|\log |\delta_{S_i}(R)|)$ time for each $i\in [k]$. Because $S_i\cap S_j=\emptyset$ for every distinct $i,j\in [k]$, \eqref{eq:equivmin} can be solved in $O(\sum_{i=1}^k(|\delta_{S_i}(R)|\log|\delta_{S_i}(R)|))=O(|\delta(R)|\log|\delta(R)|)$ time. Consequently, \eqref{eq:knapsackmin} can be solved in $O((\prod_{i=1}^k q_i)|\delta(R)|\log{|\delta(R)|})$ time. Since $q_i\leq |E|$ for every $i\in [k]$ and $|\delta(R)|\leq |E|$, the run-time is $O(|E|^{k+2})$.
\end{proof}

Lemma \ref{lem:genknap} by itself does not enable efficient separation over the family of constraints \eqref{eq:genconstraint} since the number of possible subsets $R$ is exponential in $|V|$. To address this issue, we will rely on the global min-cut value of $(G,u_x)$ to determine whether a given point $x$ satisfies all constraints of the LP-relaxation and if not, then find a violated constraint. The following lemma establishes two ranges of the global min-cut value in $(G,u_x)$ where the separation problem is easy: the first range leads to a constraint violated by $x$, and the second range certifies that $x$ satisfies all constraints of the LP-relaxation of \ref{eq:IP1}.
\begin{lemma}\label{lem:gencapacitated}
    Let $x\in [0,1]^E$. 
    We have the following:
    \begin{enumerate}
        \item Let $\emptyset \neq R\subsetneq V$ such that $u_x(\delta(R))<\prod_{i=1}^kq_i$. Then, $x$ violates constraint \eqref{eq:genconstraint} for $(J=\emptyset, R)$.
        \item Let $\emptyset \neq R\subsetneq V$ such that $u_{x}(\delta(R))\geq k\prod_{i=1}^kq_i$. Then, $x$ satisfies constraint \eqref{eq:genconstraint} for $(J, R)$ for every $J\subseteq \delta(R)$.
    \end{enumerate}
\end{lemma}
\begin{proof}
We prove the two parts below. 
    \begin{enumerate}
        \item Suppose we have $\emptyset\neq R\subsetneq V$ such that $u_x(\delta(R))<\prod_{i=1}^kq_i$. Let $J=\emptyset$. Then, using Proposition \ref{prop:capacitycut}, the LHS of constraint \eqref{eq:genconstraint} is 
        \begin{equation}
            \sum_{i=1}^k\left(\prod_{j\in [k]\backslash\{i\}}q_j\right)\cdot x(\delta_{T_i}(R))=\sum_{i=1}^kQ_i\cdot x(\delta_{T_i}(R))=u_x(\delta(R)).
        \end{equation}
        In addition, the RHS of constraint \eqref{eq:genconstraint} is 
            $\prod_{i=1}^k(q_i-|J\cap T_i|)^+=\prod_{i=1}^kq_i$. 
        Because $u_x(\delta(R))<\prod_{i=1}^kq_i$, it follows that $x$ violates constraint \eqref{eq:genconstraint} for $(J=\emptyset, R)$.
        \item Let $\emptyset\neq R\subsetneq V$ such that $u_x(\delta(R))\geq k\prod_{i=1}^kq_i$. Let $J\subseteq \delta(R)$ be a subset of $\delta(R)$. We show that $x$ satisfies constraint \eqref{eq:genconstraint} for $(J, R)$. 
        If there exists $i\in [k]$ such that $q_i\leq |J\cap T_i|$, then the RHS of \eqref{eq:genconstraint} for $(J, R)$ is equal to $0$ and hence, the constraint is satisfied due to non-negativity of all variables. Consequently, we may assume that $q_i>|J\cap T_i|$ for all $i\in [k]$. If $x(\delta_{T_i}(R))<q_i$ for all $i\in [k]$, then by Proposition \ref{prop:capacitycut}, we have that
        \begin{equation*}
            u_x(\delta(R))=\sum_{i=1}^kQ_i\cdot x(\delta_{T_i}(R))<\sum_{i=1}^k\prod_{j=1}^kq_j=k\prod_{j=1}^kq_j,
        \end{equation*}
        a contradiction to the assumption on $R$. Hence, there exists $m\in [k]$ such that $x(\delta_{T_m}(R))\geq q_m$. Then, we have that
    \begin{align*}
        \text{LHS of }\eqref{eq:genconstraint} \text{ for $(J, R)$}
        &=\sum_{i=1}^k\left(\prod_{j\in [k]\backslash\{i\}}(q_j-|J\cap T_j|)\right)\cdot x(\delta_{T_i}(R)\backslash J)\\
        &\geq \left(\prod_{j\in [k]\backslash\{m\}}(q_j-|J\cap T_j|)\right)\cdot x(\delta_{T_m}(R)\backslash J) \quad \quad \text{(since $x \ge 0$)}\\
        &= \left(\prod_{j\in [k]\backslash\{m\}}(q_j-|J\cap T_j|)\right)\cdot (x(\delta_{T_m}(R))-x(J\cap T_m))\\
        & \quad \quad \quad \quad \quad \quad \quad \quad \quad \quad \text{(since $J\subseteq \delta(R)$)}\\
        &\ge \left(\prod_{j\in [k]\backslash\{m\}}(q_j-|J\cap T_j|)\right)\cdot (x(\delta_{T_m}(R))-|J\cap T_m|) \\
        & \quad \quad \quad \quad \quad \quad \quad \quad \quad \quad \text{(since $x\le 1$)}\\
        &\geq \left(\prod_{j\in [k]\backslash\{m\}}(q_j-|J\cap T_j|)\right)\cdot (q_m-|J\cap T_m|)\\
        &=\text{RHS of }\eqref{eq:genconstraint} \text{ for $(J, R)$}.
    \end{align*}
    \end{enumerate}
\end{proof}
We now solve the separation problem for the family of constraints \eqref{eq:genconstraint} using Lemmas \ref{lem:genknap} and \ref{lem:gencapacitated}, and thereby prove Lemma \ref{thm:LPpolysolve}.
\begin{proof}[Proof of Lemma \ref{thm:LPpolysolve}] 
    We note that the number of constraints in the LP-relaxation of \ref{eq:IP1} is exponential in $|V|$. By the result of Grötschel, Lovász, and Schrijver \cite{grotschel1981ellipsoid}, in order to optimize over a polyhedron, it suffices to solve the separation problem. We will design a polynomial-time separation oracle for the LP-relaxation: given $x$, we need to verify whether $x$ satisfies all constraints in \eqref{eq:genconstraint} and if not, return a violated constraint. We use Algorithm \ref{alg:separation_oracle} for our separation oracle.
    \begin{algorithm}[H]
    \caption{Separation Oracle}
    \label{alg:separation_oracle}
    \begin{algorithmic}[1]
    \State \textbf{Input: }A vector $x\in [0,1]^E$
    \State $R^*\gets\arg\min\{u_x(\delta(R)):\emptyset \neq R\subsetneq V\}$ and $\lambda \gets u_x(\delta(R^*))$
    \If{$\lambda < \prod_{i=1}^k q_i$}
        \State \Return $(J=\emptyset, R=R^*)$ \Comment{Violated inequality found}
    \Else
        \ForAll{$\emptyset \neq R \subsetneq V$ such that $u_x(\delta(R)) \leq k\prod_{i=1}^k q_i$}
            \State $J\gets$ optimum solution of problem \eqref{eq:knapsackmin} with inputs $x$ and $R$
            \If{optimum objective value of \eqref{eq:knapsackmin} is negative}
                \State \Return $(J,R)$  \Comment{Violated inequality found}
            \EndIf
        \EndFor
    \EndIf
    \State \Return $x$ satisfies all constraints
    \end{algorithmic}
    \end{algorithm}
    We will now prove the correctness of Algorithm \ref{alg:separation_oracle}. First, suppose that $x$ satisfies constraint \eqref{eq:genconstraint} for $(J,R)$ for every $J\subseteq \delta(R)$ and $\emptyset\neq R\subseteq V$. We will show that the algorithm returns that $x$ satisfies all constraints. Because $x$ satisfies constraint \eqref{eq:genconstraint} for $(J,R)$ for every $J\subseteq \delta(R)$ and $\emptyset\neq R\subseteq V$, by the first part of Lemma \ref{lem:gencapacitated}, we have that $\lambda\geq \prod_{i=1}^k q_i$ and hence, the algorithm will not return in Step 3. Furthermore, for a fixed $\emptyset\neq R\subseteq V$, the optimum objective value of \eqref{eq:knapsackmin} is negative if and only if $x$ violates constraint \eqref{eq:genconstraint} for $(J,R)$ for some $J\subseteq R$. Because $x$ satisfies constraint \eqref{eq:genconstraint} for $(J,R)$ for every $J\subseteq \delta(R)$ and $\emptyset\neq R\subsetneq V$, the algorithm will not return in Step 8 for every choice of $R$ considered in Step 5. Therefore, the algorithm returns that $x$ satisfies all constraints.

    Next, suppose that there exists $J'\subseteq\delta(R')$ and $\emptyset \neq R'\subsetneq V$ such that $x$ violates constraint \eqref{eq:genconstraint} for $(J',R')$. We show that the algorithm returns some subset $(J,R)$ such that $x$ violates constraint \eqref{eq:genconstraint} for $(J,R)$. We have two cases based on the value of $\lambda$, where $\lambda\coloneqq \min\{u_x(\delta(R)):\emptyset\neq R\subsetneq V\}$. In the first case, suppose $\lambda<\prod_{i=1}^kq_i$. Let $R^*\in \arg \min \{u_x(\delta(R)):\emptyset\neq R\subsetneq V\}$. Then, by the first part of Lemma \ref{lem:gencapacitated}, we have that $x$ violates constraint \eqref{eq:genconstraint} for $(J=\emptyset, R^*)$, and the algorithm correctly returns this in Step 3. In the second case, suppose $\lambda\geq \prod_{i=1}^kq_i$. Because $x$ violates constraint \eqref{eq:genconstraint} for $(J',R')$, by the second part of Lemma \ref{lem:gencapacitated}, we have that $u_x(\delta(R'))<k\prod_{i=1}^kq_i$. Therefore, $R'$ is explored by the algorithm in Step 5. Furthermore, the optimum objective value of \eqref{eq:knapsackmin} with input $R'$ and $x$ is negative since $J'$ is a feasible solution to the minimization problem given by \eqref{eq:knapsackmin} with negative objective value. Thus, the optimum solution $J$ for the minimization problem given by \eqref{eq:knapsackmin} with $R$ and $x$ as inputs also has negative objective value, and hence, the tuple $(J,R')$ returned by the algorithm in Step 8 is a tuple for which $x$ violates \eqref{eq:genconstraint}.
    
    We now bound the runtime of Algorithm \ref{alg:separation_oracle}. We note that we can compute the global min-cut of the capacitated graph $(G,u_x)$ in polynomial time. Furthermore, if the min-cut is less than $\prod_{i=1}^kq_i$, then Algorithm \ref{alg:separation_oracle} terminates and returns a violated inequality. Therefore, $\lambda \geq\prod_{i=1}^kq_i$, meaning every cut $R\in \mathcal{R}$ is a $k$-approximate min-cut. By Theorem \ref{thm:kargercuts}, the number of $k$-approximate min-cuts is $O(|V|^{2k})$ and hence, the number of sets $R$ explored by the algorithm in Step 5 is $O(|V|^{2k})$. Furthermore, for each $R$ that is explored in Step 5, the minimization problem in \eqref{eq:knapsackmin} can be solved in $O(|E|^{2k+2})$ time by Lemma \ref{lem:genknap}. Hence, the run-time of the algorithm is $O(|V|^{2k}|E|^{2k+2})$, which is polynomial for constant $k$.
    
    Since there exists a polynomial-time separation oracle and the feasible region is in $[0,1]^E$ and is hence bounded, we can apply the Ellipsoid method \cite{grotschel1981ellipsoid} to solve the LP relaxation of \eqref{eq:IP1} in polynomial time.
\end{proof}

\subsection{Approximation Algorithm}\label{sec:rounding}
In order to prove Theorem 1.1, we design an independent randomized rounding algorithm for an optimum solution to the LP relaxation of $\eqref{eq:IP1}$. We show that the algorithm returns a feasible solution with cost that is $O(k^2 \log n)$-factor of the optimal cost of the LP relaxation of $\eqref{eq:IP1}$ with constant probability. Our rounding algorithm is described in Algorithm \ref{alg:edge_selection}.

\begin{algorithm}
\caption{$k$-Tier FGC Approximation Algorithm}
\label{alg:edge_selection}
\begin{algorithmic}[1]
\State \textbf{Input:} A feasible solution $x$ to the LP-relaxation of \ref{eq:IP1}
\State \textbf{Output:} A subset $F \subseteq E$ of edges 

\State $A \gets \{e \in E : (100 k^2 \log n) x_e \geq 1\}$
\State $B \gets E \setminus A$
\State $B' \gets \emptyset$

\For{each edge $e \in B$}
    \State Sample $e$ with probability $(100k^2 \log n) x_e$
    \If{$e$ is sampled}
        \State $B' \gets B' \cup \{e\}$
    \EndIf
\EndFor

\State $F \gets A \cup B'$
\State \Return $F$
\end{algorithmic}
\end{algorithm}

Let $F$ be the set of edges returned by Algorithm \ref{alg:edge_selection}. We now show that $F$ is feasible (i.e., $(V, F)$ is $(\mathcal{T}, \mathcal{Q})$-FGC) with high probability.

\begin{definition}
    A set $\emptyset\neq R\subsetneq V$ is \textbf{deficient} if $|\delta_{F}(R) \cap T_i| < q_i$ for all $i\in [k]$. 
\end{definition}

\begin{lemma}\label{lem:deficient}
    Let $\emptyset \neq R \subsetneq V$. 
    Then, $R$ is deficient with probability at most $n^{-20k}$. 
\end{lemma}
\begin{proof}
Let $Y_e$ be the indicator variable denoting whether $e \in F$. Let $A_{R} := A \cap \delta(R)$, $B_R := B \cap \delta(R)$. For $i \in [k]$, define $Z_{T_i} := \sum_{e \in B_R \cap T_i} Y_e$. Taking $J := A_{R}$, the inequality for $(J,R)$ in \ref{eq:IP1} is 
$$\sum_{i=1}^k\left(\prod_{j\in [k]\setminus \left\{i\right\}} (q_j-|A_R \cap T_j|)^+ \right)\cdot x((\delta(R)-A_R)\cap T_i)\geq \prod_{i=1}^k(q_i-|A_R\cap T_i|)^+.$$
If there exists $i\in [k]$ such that $|A_R\cap T_i|\ge q_i$, then $|\delta_F(R)\cap T_i|\ge |A_R\cap T_i|\ge q_i$, and hence, $R$ is deficient with probability zero.  Henceforth, we assume that $|A_R\cap T_i|< q_i$ for all $i\in [k]$. Then, the inequality above implies that there exists $i\in [k]$ such that $x(B_R \cap T_i) \geq \frac{1}{k} (q_i - |A_R \cap T_i|)$. 
By definition of $y$, we have that $y(B_R \cap T_i) \geq (100 k \log n) (q_i - |A_R \cap T_i|) := \mu_1$. We note that $E[Z_{T_i}] = y(B_R \cap T_i)$. So,
    \begin{align*}
    Pr[|\delta_{F}(R) \cap T_j| < q_j\ \forall j\in [k]]
    &\leq Pr[|\delta_{F}(R) \cap T_i| < q_i]\\
    &= Pr[|B_R \cap F \cap T_i| < q_i - |A_R \cap T_i|]\\
    &= Pr[Z_{T_i} < q_i - |A_R \cap T_i|]\\
    &\leq Pr[Z_{T_i} < \mu/3]\\
    &\leq \exp(-(20k \log n) (p-|A_R \cap S|)) \quad  \text{(by Chernoff bound)}\\
    &\leq \exp (-20 k\log n)\\
    &= n^{-20k}.
    \end{align*}
\end{proof}

We recall that $(G,u_{x})$ is the graph with edge capacities $u_x:E\rightarrow \R_{\ge 0}$ defined by $u_{x}(e) := \left(\sum_{j=i}^{k} \prod_{\ell \in [k] \setminus \left\{j\right\}} q_{\ell}]\right)\cdot x_e$ for each $e \in S_i$ for each $i\in [k]$. We now upper bound the probability that high capacity cuts in $(G,u_{x})$ are deficient.
\begin{lemma}\label{lem:deficient_big}
    Let $\emptyset\neq R\subsetneq V$ such that $u_x(\delta(R)) \geq \ell \prod_{i=1}^{k} q_i$ for some $\ell\ge 2k$. Then, $R$ is deficient with probability at most $n^{-10\ell}$. 
\end{lemma}
\begin{proof}
    We assume $|A_R \cap T_i| < q_i$ for all $i \in [k]$ because the claim holds otherwise. By Proposition \ref{prop:capacitycut}, we have that $u_x(\delta(R)) = \sum_{i=1}^{k} \left(\prod_{j \in [k] \setminus \left\{i\right\}} q_j\right)\cdot x(\delta(R) \cap T_i)$. Since $u_x(\delta(R)) \geq \ell \cdot  \prod _{i=1}^k q_i$, there exists an index $j \in [k]$ such that $x((\delta(R) \cap T_j)) \geq \frac{\ell q_j }{k}$. Since $x(A_R \cap T_j) < q_j$ and $\ell \geq 2k$, we have that $x(B_R \cap T_j) \geq \frac{\ell q_j}{2k}$. This implies that $E[Z_{T_j}] = y(B_R \cap T_j) \geq 50 \ell \log n \cdot q_j$. Let $\mu_1 = 50 \ell \log n \cdot q_j$. Thus, 
    \begin{align*}
        Pr[X_R = 1] &\leq Pr[|\delta_{F}(R) \cap T_j| < q_j]\\
        &\leq Pr[|B_R \cap T_j \cap F| < q_j]\\
        &\leq Pr[Z_{T_j} \leq q_j]\\
        &\leq Pr[Z_{T_j} < \mu_1/3]\\
        &\leq \exp(-10 \cdot \ell \log n \cdot p_j)\\
        &\leq \exp(-10 \cdot \ell \log n)\\
        &= n^{-10 \ell}.
    \end{align*}
\end{proof}
Having established upper bounds on the probability that a given subset $R \subsetneq V$ is deficient, we now show that none of the sets are deficient with high probability.
\begin{lemma}\label{lem:high_prob}
    Algorithm \ref{alg:edge_selection} returns a subset $F \subseteq E$ such that $(V,F)$ is feasible with probability $1-n^{-O(1)}$. 
\end{lemma}
\begin{proof}
For $F$ to be infeasible, there must be a deficient cut $\emptyset\subseteq R\subsetneq V$. 
We partition the nontrivial cuts in the capacitated graph of $H_x:=(G,u_{x})$ as $\mathcal{C}_{<2k} \cup \mathcal{C}_{2k} \cup \mathcal{C}_{2k+1} \cup \dots$, where 
\begin{align*}
\mathcal{C}_{<2k} &:= \left\{\emptyset \neq R\subsetneq V: u_x(\delta(R)) \in \left[\prod_{i=1}^{k} q_i, 2k\cdot\prod_{i=1}^{k} q_i\right)\right\} \text{ and}    \\
\mathcal{C}_\ell &:= \left\{\emptyset \neq R\subsetneq V: u_x(\delta(R)) \in \left[\ell \cdot\prod_{i=1}^{k} q_i, (\ell + 1) \cdot\prod_{i=1}^{k} q_i\right)\right\} \forall\ \ell\in \{2k, 2k+1, \ldots\}. 
\end{align*}

 By substituting $J = \emptyset$ into constraint \eqref{eq:genconstraint}, we obtain that the minimum cut in $H_{x}$ is at least $\prod_{i=1}^{k} q_i$. Using Theorem \ref{thm:kargercuts}, we have that $|\mathcal{C}_{<2k}| = O(n^{4k})$. For $\emptyset\neq R\subsetneq V$, let $Z_R$ be the random variable that indicates if $R$ is deficient. By Lemma \ref{lem:deficient}, we have that $\mathbb{P}[Z_R = 1] \leq n^{-20k}$ for every subset $R \subsetneq V$. Thus, 
\[
\sum_{R \in \mathcal{C}_{<2k}} \mathbb{P}[Z_R = 1] = O(n^{-16k}).
\]
For $\ell \in \left\{2k,2k+1,2k+2,\ldots \right\}$, Theorem \ref{thm:kargercuts} implies that $|\mathcal{C}_\ell| = O(n^{2\ell + 2})$. By Lemma \ref{lem:deficient_big}, we have that $\mathbb{P}[Z_R = 1] \leq n^{-10\ell}$ for each $R \in \mathcal{C}_\ell$. Thus, for every integer $\ell \geq 2k$, we have
\[
\sum_{R \in \mathcal{C}_\ell} \mathbb{P}[Z_R = 1] = O(n^{-8\ell + 2}) = O(n^{-2\ell}).
\]
Overall, we get that 
\begin{align*}
\mathbb{P}[\text{$F$ is infeasible}] &= \sum_{R:\ \emptyset \neq R\subsetneq V} \mathbb{P}[Z_R = 1]\\ 
&= \sum_{R \in \mathcal{C}_{<2k}} \mathbb{P}[Z_R = 1] + \sum_{\ell=2k}^\infty \sum_{R \in \mathcal{C}_\ell} \mathbb{P}[Z_R = 1]\\
&= O(n^{-2}) + \sum_{\ell=2k}^\infty O(n^{-2\ell}) = n^{-O(1)}.
\end{align*}
\end{proof}

\begin{lemma}\label{lem:rounding}
    Algorithm \ref{alg:edge_selection} returns a subset $F$ such that with constant probability, $F$ is feasible and 
    $c(F) \leq (200k^2 \log n) \sum_{e\in E}c_e x_e$.
\end{lemma}

\begin{proof}
    Let cost$(x):=\sum_{e\in E} c_e x_e$ and 
    $F$ be the set of edges returned by Algorithm 2. Then $E[c(F)] \leq 100k^2 \log n$ cost$(x)$. By, Markov's Inequality, we have that $\mathbb{P}(c(F) \geq (200 k^2 \log n) \text{cost}(x)) \leq \mathbb{P}(c(F) \geq 2 E[c(F)]) \leq \frac{1}{2}$. By Lemma \ref{lem:high_prob}, $F$ is feasible w.h.p, which implies that $F$ is a feasible solution to $k$-tier FGC and $c(F) \leq (200 k^2 \log n) \text{cost}(x)$ with probability at least $\frac{1}{3}$ for $n$ sufficiently large.
    \end{proof}

\subsection{Strength of the LP-relaxation of \ref{eq:IP1}}\label{sec:LP-strength}
We will compare the LP-relaxation of our IP for $(p, q)$-FGC to that of Ibrahimpur-V\'{e}gh's IP. Next, we show that the LP-relaxation of \ref{eq:IP1} has an integrality gap of at least $2k$. 
    
  Our IP formulation \ref{eq:IP1} for $k=2$ closely resembles the IP formulation of Ibrahimpur and V\'{e}gh for $(p, q)$-FGC \cite{IV25}. 
  For completeness, we explicitly state \ref{eq:IP1} for $k=2$ below (where $p=q_1$, $q=q_2 - q_1$, $S=T_1$ and $U=T_2\setminus T_1$): 
\begin{align*}\tag{$2$-tier-IP}\label{eq:IP2}
    \min & \sum_{e\in E}c_e x_e\\
    (p-|J\cap S|)^+ \cdot x(\delta(R)-J) &+ (p + q - |J|)^+\cdot x(\delta_S(R)-J) \\
    & \quad \quad \quad \ge (p-|J\cap S|)^+\cdot (p+q-|J|)^+\ \forall\ J\subseteq \delta(R), \emptyset \neq R\subsetneq V \stepcounter{equation}\\ 
    x_e &\in \left\{0,1\right\} \ \ \forall\ e\in E\\
\end{align*}
The IP formulation given by Ibrahimpur and V\'{e}gh \cite{IV25} is the following: 
\begin{align*}\tag{$2$-tier-IV-IP}\label{eq:IP3}
    \min & \sum_{e\in E}c_e x_e\\
    (p-|J\cap S|)^+\cdot  x(\delta(R)-J) &+ (q-|J \cap U|)^+\cdot x(\delta_S(R)-J) \\
    & \quad \quad \quad \ge (p-|J\cap S|)^+\cdot (p+q-|J|)^+\ \forall\ J\subseteq \delta(R), \emptyset \neq R\subsetneq V \stepcounter{equation}\\ 
    x_e &\in \left\{0,1\right\} \ \ \forall\ e\in E\\
\end{align*}
We note that the only difference between \ref{eq:IP2} and \ref{eq:IP3} is that for \ref{eq:IP2}, the coefficient for the $x((\delta(R) - J) \cap S)$ term is $(p+q - |J|)^{+}$ whereas for \ref{eq:IP3}, the corresponding coefficient is $(q-|J \cap U|)^{+}$. 
\begin{proposition}
    The LP relaxation of \ref{eq:IP2} is weaker than the LP relaxation of \ref{eq:IP3}. Moreover, there exists an instance where the objective value of the LP relaxation of \ref{eq:IP2} is strictly smaller than the LP relaxation of \ref{eq:IP3}. 
\end{proposition}
\begin{proof}
    Let $x$ be a feasible solution to the  the LP relaxation of \ref{eq:IP3}. We show that $x$ is a feasible solution to the LP relaxation of \ref{eq:IP2}. For $\emptyset \neq R \subsetneq V$ and $J \subseteq \delta(R)$, consider the constraint in \ref{eq:IP2} corresponding to $(J, R)$. If $p < |J \cap S|$ then $(p-|J \cap S|)^{+} \cdot(p+q-|J|)^{+} = 0$ and hence, $x$ satisfies the constraint in \ref{eq:IP2} for $(J, R)$. So, we assume $p > |J \cap S|$. 

    If $q > |J \cap U|$, then $(p+q-|J|)^{+} = (p-|J \cap S|) + (q-|J \cap U|) > (q - |J \cap U|)^{+}$ and hence, $x$ satisfies the constraint in \ref{eq:IP2} for $(J, R)$. If $q \leq |J \cap U|$, then we have $(p + q - |J|)^{+} \geq 0 = (q-|J \cap U|)^{+}$, and hence, $x$ satisfies the constraint in \ref{eq:IP2} for $(J, R)$. 

We now exhibit an instance where the objective value of the LP relaxation of \ref{eq:IP2} is strictly smaller than the LP relaxation of \ref{eq:IP3}. Consider the instance $G = (V,E)$ with $V = \left\{v_1, v_2\right\}$, $E = S = \left\{\{v_1,v_2\}\right\}$, $U=\emptyset$, $p = 1$, $q=0$, and cost $c(\{v_1, v_2\}) = 1$, . The only non-trivial constraint for \ref{eq:IP2} and \ref{eq:IP3} corresponds to $(J = \emptyset, R = \left\{v_1\right\})$. For \ref{eq:IP2}, this constraint is  $x_{\left\{v_1,v_2\right\}} \geq \frac{1}{2}$ and for \ref{eq:IP3}, this constraint is  $x_{\left\{v_1,v_2\right\}} \geq 1$. Thus, the LP relaxation of \ref{eq:IP2} has optimal cost $\frac{1}{2}$ while the LP relaxation of \ref{eq:IP3} has optimal cost $1$.
\end{proof}

We now establish a lower bound of $2k$ on the integrality gap of the LP-relaxation of \ref{eq:IP1}.

\begin{lemma}\label{lem:gap_2}
    The LP relaxation of \ref{eq:IP1} has an integrality gap of at least $2k(1-1/|V|)$.
\end{lemma}

\begin{proof}
    Consider the instance  with graph $G = C_n$, tiers $T_1 = T_2 = \cdots = T_k = E$, costs  $c: E \rightarrow \mathbb{R}_{\geq 0}$ given by $c(e) = 1$ for all $e \in E$, and tier requirements $q_1 = q_2 = \cdots = q_k = 1$. The optimal solution to IP \ref{eq:IP1} is a spanning path and hence has cost $n-1$. We show below that $x_e=1/(2k)$ for every $e\in E$ is a feasible solution to the LP-relaxation of \ref{eq:IP1} and hence, the optimal objective value of the LP-relaxation is at most $n/2k$. This leads to the stated integrality gap. 
    
    We now show that $x$ is feasible to the LP-relaxation of \ref{eq:IP1}. For all $j \in [k]$, we observe that $(q_j - |J \cap T_j|)^{+} = 0$ for $|J| \geq 1$. Thus, if $|J| \geq 1$ and $\emptyset \neq R \subsetneq V$, inequality (\ref{eq:genconstraint}) corresponding to $(J, R)$ is trivially satisfied. Thus, the only non-trivial inequalities in (\ref{eq:genconstraint}) are of the form $(J=\emptyset, R)$. Fix a subset $\emptyset  \neq R \subsetneq V$ and let $ \mathcal{F}$ be the family of connected components in $G[R]$. Each connected component of $G[R]$ has a unique pair of edges in $\delta(R)$. The inequality corresponding to $(J=\emptyset, R)$ holds since    
    \begin{equation*}
        \sum_{i=1}^{k} x(\delta_{T_i}(R)) = k \cdot x (\delta(R)) \ge 2k|\mathcal{F}|/2k\geq 1.
    \end{equation*}

\end{proof}

\section{Min Cardinality $k$-tier FGC}
In this section, we design a $(\alpha_{q_1}+2(q_k-q_1)/q_1)$-approximation for min-cardinality $k$-tier FGC and prove Theorem \ref{thm:cardinality}.
We will first define a series of useful constructs that will be used in our algorithm.
\begin{definition}[Set Families and Covers]\label{def:setfamiliescovers}
    Let $G=(V,E)$ with nested family $\mathcal{T}= (T_1,T_2,\ldots,T_k)$ and requirements $\mathcal{Q}= (q_1,q_2,\ldots,q_k)\in \mathbb{Z}_{\ge 0}^k$, where $T_1\subseteq T_2\subseteq \ldots\subseteq T_k=E$ and $q_1\leq q_2\leq \ldots\leq q_k$. Let $F\subseteq E$. 
    \begin{enumerate}
        \item We define $\mathcal{Q}_{q_1}^1\coloneqq (q_1,q_1\ldots, q_1)$ and for every $i\in [k]\backslash \{1\}$ and $j\in [q_i-q_{i-1}]$, we define the following:
        \begin{align*}
            \mathcal{Q}^i_j&\coloneqq (q_1, q_2, \ldots, q_{i-1}, q_{i-1}+j, q_{i-1} + j, q_{i-1} + j,\ldots, q_i+j) \text{ and}\\
            H_j^i(F)&\coloneqq \left\{\emptyset\neq R\subsetneq V:\text{$R$ is $(\mathcal{T},\mathcal{Q}^i_{j-1})$-safe in $(V,F)$ and $(\mathcal{T},\mathcal{Q}^i_{j})$-unsafe in $(V,F)$}\right\}.
        \end{align*}
        \item Let $\mathcal{C}\subseteq 2^V\setminus \{\emptyset,V\}$. An \emph{$F$-excluding $\mathcal{C}$-cover} is a subset $C\subseteq E\setminus F$ with $|\delta(R)\cap C|\geq 1$ for every $R\in \mathcal{C}$. An inclusion-wise minimal $F$-excluding $\mathcal{C}$-cover is an $F$-excluding $\mathcal{C}$-cover $C$ if $C\setminus \{e\}$ is not an $F$-excluding  $\mathcal{C}$-cover for every $e\in C$. 
    \end{enumerate}

\end{definition}
We have the following observations based on the above definitions.
\begin{proposition}\label{prop:mincoverforest}
    Let $G=(V, E)$, $F\subseteq E$, and $\mathcal{C}\subseteq 2^V\setminus \{\emptyset,V\}$. Every inclusion-wise minimal $F$-excluding $\mathcal{C}$-cover is acyclic.
\end{proposition}
\begin{proof}
    Let $C$ be an inclusion-wise minimal $F$-excluding $\mathcal{C}$-cover. For the sake of contradiction, suppose $C$ contains a cycle $C'$. Fix an arbitrary edge $e$ in the cycle $C'$. Then, for every $R\in \mathcal{C}$ such that $e\in \delta(R)$, there exists $e'\in C'-e$ such that $e'\in \delta(R)$. Thus, $C-e$ is also an $F$-excluding $\mathcal{C}$-cover, a contradiction to the inclusion-wise minimal property of $C$.
\end{proof}

\begin{lemma}\label{lem:violated} 
   Let $G=(V,E)$ with nested family $\mathcal{T}= (T_1,T_2,\ldots,T_k)$ and requirements $\mathcal{Q}= (q_1,q_2,\ldots,q_k)\in \mathbb{Z}_{\ge 0}^k$, where $T_1\subseteq T_2\subseteq \ldots\subseteq T_k=E$ and $q_1\leq q_2\leq \ldots\leq q_k$. Let $F\subseteq E$. For every $R\in H_j^{i}(F)$, the following two properties hold:
    \begin{enumerate}
        \item $|\delta_{T_a}(R)\cap F|< q_a$ for all $a\in [i-1]$, $|\delta_{T_b}(R)\cap F|\le q_{i-1}+j-1$ for all $b\in \{i, i+1, \ldots, k\}$, 
        \item $|\delta(R)\cap F|=q_{i-1}+j-1$ and moreover, there exists $b\in \{i,i+1,\ldots, k\}$ such that $|\delta_{T_b}(R)\cap F|=q_{i-1}+j-1$.
    \end{enumerate}
\end{lemma}
\begin{proof}
    Let $R\in H_j^i(F)$. 
    Since $R$ is $(\mathcal{T},\mathcal{Q}_{j}^i)$-unsafe in $(V,F)$, we have that $|\delta_{T_a}(R)\cap F|<q_a$ for all $a\in [i-1]$ and $|\delta_{T_b}(R)\cap F|<q_{i-1}+j$ for all $b\in \{i,i+1,\ldots, k\}$. This implies the first property. Also, $R$ is $(\mathcal{T},\mathcal{Q}_{j-1}^i)$-safe in $(V,F)$. Since $|\delta_{T_a}(R)\cap F|<q_a$ for all $a\in [i-1]$, it follows that there exists $b\in \{i,i+1,\ldots,k\}$ such that $|\delta_{T_b}(R)\cap F|\geq q_{i-1}+j-1$. Since $T_1\subseteq T_2\subseteq\ldots\subseteq T_k=E$, this implies that $q_{i-1}+j-1=|\delta_{T_b}(R)\cap F|\leq |\delta_{T_k}(R)\cap F|<q_{i-1}+j$, and hence, $|\delta(R)\cap F|=q_{i-1}+j-1$.
\end{proof}

We now show that for a subset $F\subseteq E$ such that $(V, F)$ is $(\mathcal{T}, \mathcal{Q}^i_{j-1})$-FGC, all sets in the family $H^i_j(F)$ can be enumerated in polynomial time and moreover an $F$-excluding $H^{i}_{j}(F)$ cover can be computed in polynomial time.
    
\begin{proposition}\label{lem:polyviolated}
    There exists an algorithm that 
    takes as input a graph $G=(V,E)$ with nested family $\mathcal{T}= (T_1,T_2,\ldots,T_k)$, requirements $\mathcal{Q}= (q_1,q_2,\ldots,q_k)\in \mathbb{Z}_{\ge 0}^k$, where $T_1\subseteq T_2\subseteq \ldots\subseteq T_k=E$ and $q_1\leq q_2\leq \ldots\leq q_k$ for some fixed constant $k$, $i\in [k]\setminus\{1\}$, $j\in [p_i]$, and $F\subseteq E$ such that $(V, F)$ is $(\mathcal{T}, \mathcal{Q}^i_{j-1})$-FGC, and 
    runs in polynomial time to verify if there exists an $F$-excluding $H^i_j(F)$-cover and if so, compute a minimal $F$-excluding $H^i_j(F)$-cover.
\end{proposition}
\begin{proof}
    We will show below that the family $H^i_j(F)$ can be enumerated in time $O(|E|^2|V|^{2k})$. After enumerating all these sets, we can verify whether $E\setminus F$ is an $F$-excluding $H^i_j(F)$-cover in time $O(|E|^3|V|^{2k})$ by verifying whether $|\delta(R)\cap (E\setminus F)|\ge 1$ for every $R\in H^i_j(F)$. If so, then we can compute an inclusionwise minimal $F$-excluding $H^i_j(F)$-cover by iteratively verifying whether each edge of $E\setminus F$ is necessary to maintain the $F$-excluding $H^i_j(F)$-cover property. The time per iteration is $O(|E|^3|V|^{2k})$ and hence, the overall time is $O(|E|^4|V|^{2k})$. 

    We now show that the family $H^i_j(F)$ can be enumerate in time $O(|E||V|^{2k})$. We consider the graph $G$ with $k$ edge-capacities $u_1, u_2, \ldots, u_k:E\rightarrow \mathbb{R}_{\ge 0}$ given by 
    \[
    u_a(e) \coloneqq
    \begin{cases}
        \frac{1}{q_a} &\text{ if } e\in T_a\cap F, a\in [i-1],\\
        \frac{1}{q_{i-1}+j-1} &\text{ if }e\in T_a\cap F, a\in \{i, i+1, \ldots, k\},\\
        0 &\text{ otherwise.}
    \end{cases}
    \]
    We now show the following two properties:
    \begin{enumerate}
        \item the multiobjective min-cut value in graph $G$ with edge capacities $u_1, u_2, \ldots, u_k$ is at least $1$, i.e., 
    $\min_{\emptyset \neq R\subsetneq V}\max_{a\in [k]} u_a(\delta(R)) \ge 1$ and 
        \item if $R\in H^i_j(F)$, then $R$ is a multiobjective min-cut in graph $G$ with edge capacities $u_1, u_2, \ldots, u_k$ with the objective value being exactly $1$, i.e., $\max_{a\in [k]} u_a(\delta(R)) = 1$. 
    \end{enumerate}
    Firstly, suppose $\emptyset\neq R\subsetneq V$ such that $R\not\in H^i_j(F)$. Since $(V, F)$ is $(\mathcal{T}, \mathcal{Q}^i_{j-1})$-FGC but $R\not\in H^i_j(F)$, it follows that $R$ is $(\mathcal{T}, \mathcal{Q}^i_{j})$-safe in $(V,F)$. Hence, either there exists $a\in [i-1]$ such that $u_a(\delta(R))\ge 1$ or there exists $b\in \{i, i+1, \ldots, k\}$ such that $u_b(\delta(R))>1$ and consequently, $\max_{a\in [k]} u_a(\delta(R)) \ge 1$. 
    Next, suppose $R\in H^i_j(F)$. By Lemma \ref{lem:violated}, we have that $u_a(\delta_{T_a}(R))<1$ for all $a\in [i-1]$, $u_b(\delta_{T_b}(R))\le 1$ for all $b\in \{i, i+1, \ldots, k\}$, and $u_k(\delta_{T_k}(R))=1$, and hence, $\max_{a\in [k]} u_a(\delta(R)) = 1$.  From these two, it follows that $\min_{\emptyset \neq R\subsetneq V}\max_{a\in [k]} u_a(\delta(R)) \ge 1$ and moreover, if $R\in H^i_j(F)$, then $\max_{a\in [k]} u_a(\delta(R)) = 1$. 

    Therefore, in order to enumerate $H^i_j(F)$, we verify whether $\min_{\emptyset \neq R\subsetneq V}\max_{a\in [k]} u_a(\delta(R)) = 1$ and if so, then 
    it suffices to enumerate all multiobjective min-cuts in graph $G$ with edge capacities $u_1, u_2, \ldots, u_k$ and verify whether $R\in H^i_j(F)$ for each such enumerated cut $R$. By Lemma \ref{lem:multiobj-min-cut}, the number of multiobjective min-cuts in a graph $G$ with $k$ edge capacity functions is $O(|V|^{2k})$ and they can all be enumerated in time $O(|E||V|^{2k})$. The time to verify whether a given $\emptyset\neq R\subsetneq V$ is in $H^i_j(F)$ is $O(k|E|)$. Hence, the time to enumerate $H^i_j(F)$ is $O(|E|^2|V|^{2k})$. 
\end{proof}

We now prove Theorem \ref{thm:cardinality}.
\begin{proof}[Proof of Theorem \ref{thm:cardinality}]
We consider the following algorithm:
    \begin{algorithm}[H]
    \caption{Unit Cost $k$-Tier FGC Algorithm}
    \label{alg:Unit-Cost}
    \begin{algorithmic}[1]
    \State \textbf{Input: } $G=(V,E)$, $T_1\subseteq T_2\subseteq\ldots\subseteq T_k=E$, $0=q_0\leq q_1\leq q_2\leq\ldots\leq q_k$
    \State \textbf{Output: } A subset $F\subseteq E$ of edges
    \State $C^1_{q_1}\gets \alpha_{q_1}$-approximate solution to max-sized $q_1$-ECSS on instance $(V, E)$
    \State $F\gets C^1_{q_1}$
    \For{$i=2$ to $k$}
        \For{$j=1$ to $q_i-q_{i-1}$}
            \State $C^i_j\gets$ inclusion-wise minimal $F$-excluding $H_j^i(F)$-cover
            \State $F\gets F\cup C^i_j$
        \EndFor
    \EndFor
    \State \Return $F$
    \end{algorithmic}
    \end{algorithm}
    
We note that $F$-excluding $H^i_j(F)$-cover exists for every choice of $F$ since we may assume that the input is feasible and hence, $G$ itself is $(\mathcal{T}, \mathcal{Q})$-FGC.

We now prove the correctness of Algorithm \ref{alg:Unit-Cost} by showing that $(V, F)$ is $(\mathcal{T}, \mathcal{Q})$-FGC. For each $i\in [k]\backslash \{1\}$, we define $C_0^i:=C_{q_{i-1}-q_{i-2}}^{i-1}$. 
First, we show for each $i\in [k]$, after adding $C^i_{q_i-q_{i-1}}$ to $F$, the subgraph $(V,F)$ is $(\mathcal{T},\mathcal{Q}^i_{q_i-q_{i-1}})$-FGC by induction on $i$. For the base case of $i=1$, $C^1_{q_1}$ is a $q_1$-ECSS subgraph of $(V, E)$, which is $(\mathcal{T},\mathcal{Q}_{q_1}^{1})$-FGC. Suppose the induction hypothesis is true for $i\geq 1$. We prove the inductive statement for $i+1$ next. 

For this, we prove that for each $j\in \{0\}\cup[q_{i+1}-q_i]$, after adding $C_j^{i+1}$ to $F$, the subgraph $(V,F)$ is $(\mathcal{T},\mathcal{Q}_j^{i+1})$-FGC. We will prove this by induction on $j$. For the base case of $j=0$, we have that after adding $C^{i+1}_0$ to $F$ with $C^{i+1}_0=C_{q_i-q_{i-1}}^i$, the subgraph $(V,F)$ is $(\mathcal{T},\mathcal{Q}_{q_i-q_{i-1}}^i)$-FGC by the induction hypothesis on $i$, which is equivalent to the subgraph $(V,F)$ being $(\mathcal{T},\mathcal{Q}_0^{i+1})$-FGC by part 1 of definition \ref{def:setfamiliescovers}. Suppose the induction hypothesis is true for $j\geq 0$. We prove the inductive statement for $j+1$ next. After adding $C_j^{i+1}$ to $F$ from the inductive hypothesis on $j$, the subgraph $(V,F)$ is $(\mathcal{T},\mathcal{Q}_j^{i+1})$-FGC. If the subgraph $(V,F)$ is $(\mathcal{T},\mathcal{Q}_{j+1}^{i+1})$-FGC, then we are done. Suppose that the subgraph $(V,F)$ is not $(\mathcal{T},\mathcal{Q}_{j+1}^{i+1})$-FGC. Consequently, there exists $\emptyset\neq R\subsetneq V$ such that $R$ is $(\mathcal{T},\mathcal{Q}_{j+1}^{i+1})$-unsafe in $(V,F)$, and hence $H_{j+1}^{i+1}(F)$ is nonempty. Let $R\in H_{j+1}^{i+1}(F)$. Then, by the second part of Lemma \ref{lem:violated}, $|\delta(R)\cap F|=q_i+j$. In addition, by construction, $|\delta(R)\cap C_{j+1}^{i+1}|\geq 1$, meaning $|\delta(R)\cap (F\cup C_{j+1}^{i+1})|\geq q_i+j+1$. Therefore, $R$ is $(\mathcal{T},\mathcal{Q}_{j+1}^{i+1})$-safe in $(V,F\cup C_{j+1}^{i+1})$, and this is true for all $R\in H_{j+1}^{i+1}(F)$, meaning the graph $(V, F\cup C_{j+1}^{i+1})$ is $(\mathcal{T},\mathcal{Q}_{j+1}^{i+1})$-FGC, completing the inner induction on $j$. 
Applying this conclusion for $j=q_{i+1}-q_i$ proves the outer induction statement for $i+1$. 

We now bound the approximation factor of Algorithm \ref{alg:Unit-Cost}. Let $F^*\subseteq E$ be an optimum solution and let $OPT = |F^*|$. Because $(V,F^*)$ must also be $q_1$-ECSS, for each $v\in V$, we have that $deg_{F^*}(v)\geq q_1$. Thus, by the handshaking lemma, $OPT\geq \frac{q_1|V|}{2}$. Because $C_{q_1}^1$ is an $\alpha_{q_1}$-approximate solution to max-sized $q_1$-ECSS on $(V,E)$, $|C_{q_1}^1|\leq \alpha_{q_1}\cdot OPT$. For all $i\in [k]\setminus \{1\}$ and $j\in [q_i-q_{i-1}]$, since $C_j^i$ is forest by Proposition \ref{prop:mincoverforest}, we have that 
\begin{equation*}
    |C_j^i|\leq |V|-1<|V|\leq \frac{2OPT}{q_1}.
\end{equation*}
Thus, summing over all $i\in [k]\setminus \{1\}$ and $j\in [q_i-q_{i-1}]$, we obtain
\begin{equation*}
    |F|=|C_1|+\sum_{i=2}^k\sum_{j=1}^{q_i-q_{i-1}}|C_j^i|<\alpha_{q_1} \cdot OPT+\sum_{i=2}^k\frac{2(q_i-q_{i-1})}{q_1}\cdot OPT = \left(\alpha_{q_1} + \frac{2(q_k-q_1)}{q_1}\right)\cdot OPT,
\end{equation*}
giving us the desired approximation ratio.

By Proposition \ref{lem:polyviolated}, finding an inclusion-wise minimal $F$-excluding $H_j^i(F)$-cover can be done in polynomial-time for each $i\in [k]\setminus \{1\}$ and $j\in [q_i-q_{i-1}]$. In addition, finding an $\alpha$-approximate $q_1$-ECSS also takes polynomial-time, thus proving the result.
\end{proof}

\section{$k$-tier FMGC}
In this section, we obtain a $2$-approximation for $k$-tier FMGC and prove Theorem \ref{thm:multi-use}. We recall that in $k$-tier FMGC, the input is a graph $G=(V, E)$ with non-negative edge costs $c: E\rightarrow \R_{\ge 0}$, a nested family of edge-tiers $T_1\subseteq T_2 \subseteq \ldots T_k = E$, and non-negative integral tier requirements $q_1 \le q_2 \le \ldots q_k$. The goal is to find a non-negative integral vector $x\in \Z^E_{\ge 0}$ minimizing $\sum_{e\in E}c_e x_e$ such that for every $\emptyset\neq R\subsetneq V$, there exists $i\in [k]$ with $x(\delta(R)\cap T_i)\ge q_i$. We assume throughout that $q_1>0$, otherwise $x=0$ is an optimum solution.  

For an input instance $(G=(V, E), c: E\rightarrow \R_{\ge 0}, \mathcal{T}=(T_1, T_2, \ldots, T_k), \mathcal{Q}=(q_1, q_2, \ldots, q_k))$, let $S_i:=T_i\setminus T_{i-1}$ for every $i\in \{1, 2, \ldots, k\}$ with $T_0:=\emptyset$. We observe that $S_1, \ldots, S_k$ is a partition of $E$. We consider the following IP: 
\begin{align*}\tag{IP$_{\text{FMGC-rel}}$}\label{eq:IP-FMGC-rel}
        \min \quad &\sum_{e\in E}c_ex_e\\
        \text{s.t.}\quad &\sum_{i=1}^k\frac{x(\delta_{S_i}(R))}{q_i}\geq 1
        &\forall \emptyset\neq R\subsetneq V\\
        &x_e\geq 0
        &\forall e\in E\\
        &x\in \mathbb{Z}^E.
    \end{align*}

We note that \ref{eq:IP-FMGC-rel} does not formulate $k$-tier FMGC (see Simmons' thesis \cite{simmons-thesis} for an example for the case of $k=2$). However, we show that every feasible solution to $k$-tier FMGC is feasible for \ref{eq:IP-FMGC-rel}---see Lemma \ref{lem:FMGC-feasibility}. 
\begin{lemma}\label{lem:FMGC-feasibility}
    Let $x\in \Z^E$ be a feasible solution for $k$-tier FMGC. Then, $x$ is feasible to \eqref{eq:IP-FMGC-rel}. 
\end{lemma}
\begin{proof}
    Let $\emptyset\neq R\subsetneq V$. Since $x$ is feasible for $k$-tier FMGC, there eixsts $j\in [k]$ such that $x(\delta_{T_j}(R))\ge q_j$. Since 
    \begin{align*}
        T_1\subseteq T_2\subseteq &\ldots \subseteq T_k \quad\text{and}\\
        q_1\leq q_2\leq &\ldots \leq q_k, 
    \end{align*}
    we have that, 
    \begin{align*}
        \sum_{i=1}^k\frac{x(\delta_{S_i}(R))}{q_i}\geq \frac{x(\delta_{T_j}(R))}{q_j}+\sum_{i=j+1}^k\frac{x(\delta_{S_i}(R))}{q_i}\geq \frac{x(\delta_{T_j}(R))}{q_j} \ge 1.
    \end{align*}
\end{proof}

We have the following proposition concerning the LP-relaxation of \ref{eq:IP-FMGC-rel}.
\begin{proposition}\label{prop:FMGC-lp-rel}
    The LP-relaxation of \ref{eq:IP-FMGC-rel} can equivalently be written as follows:
    \begin{align*}\tag{$\text{LP}_{\text{FMGC-rel}}$}\label{eq:LP-FMGC-rel}
        \min \quad &\sum_{i=1}^k\sum_{e\in S_i}(q_i\cdot c_e)y_e\\
        y(\delta(R))&\geq 1\ \quad\forall\ \emptyset\neq R\subsetneq V\\
        y_e&\geq 0\ \quad\forall\ e\in E.
    \end{align*}
    Moreover, for $w_e:=q_i c_e$ for each $e\in S_i$ for each $i\in [k]$, the minimum $w$-weight of a spanning tree in $G$ is at most twice the optimum objective value of the above LP. 
\end{proposition}
\begin{proof}
The LP-relaxation of \ref{eq:IP-FMGC-rel} is equivalent to \ref{eq:LP-FMGC-rel} by a transformation of variables: set $y_e = x_e/q_i$ for each $e\in S_i$ for each $i\in [k]$.  Moreover, for $w_e:=q_i c_e$ for each $e\in S_i$ for each $i\in [k]$, the minimum $w$-weight of a spanning tree in $G$ is at most twice the optimum objective value of the above LP since the integrality gap of the cut-based LP-relaxation for min-cost spanning tree is at most $2$ \cite[Theorem 22.9]{vazirani2001approximation}. 

\end{proof}

We now prove Theorem \ref{thm:multi-use}. 
\begin{proof}[Proof of Theorem \ref{thm:multi-use}]
We consider the following algorithm. Let $w_e:=q_i c_e$ for each $e\in S_i$ for each $i\in [k]$. Find a minimum $w$-weight spanning tree $T'$ in $G$ and return the vector $x\in \Z^E$, where 
\[
x_e\coloneqq
\begin{cases}
    q_i &\text{ if }e\in S_i\cap T' \text{ for some } i \in [k],\\
    0 &\text{ if }e\in E\setminus T'.
\end{cases}
\]
The runtime of the algorithm is the time to find a min-weight spanning tree which is polynomial. Next, we show that the vector $x$ returned by the algorithm is feasible for $k$-tier FMGC. Let $\emptyset \neq R\subsetneq V$. Since $T'$ is a spanning tree, there exists $e\in T'\cap \delta(R)$. Without loss of generality, let $e\in S_i$. Then, by definition $x_e = q_i$. Therefore, $x(\delta_{T_i}(R))\ge x_e = q_i$ and hence, $x$ is feasible. 

Next, we bound the cost of the solution $x$. 
Let $\text{OPT}$ denote the optimum objective value of the instance, $\text{IP-opt}_{\text{FMGC-rel}}$ denote the optimum objective value of \ref{eq:IP-FMGC-rel} for the same instance, and $\text{LP-opt}_{\text{FMGC-rel}}$ denote the optimum objective value of the LP-relaxation of \ref{eq:IP-FMGC-rel} for the same instance. Then, 
\begin{align}
        \sum_{e\in E}c_ex_e&=\sum_{i=1}^k\sum_{e\in S_i}c_ex_e\\
        &=\sum_{i=1}^k\sum_{e\in S_i}(q_i\cdot c_e)1_{e\in F} \notag\\
        &\le 2\text{LP-opt}_{\text{FMGC-rel}} \quad \quad \text{(by Proposition \ref{prop:FMGC-lp-rel})} \notag\\
        &\leq 2\text{IP-opt}_{\text{FMGC-rel}} \notag\\
        &\leq 2\text{OPT}. \quad \quad \text{(by Lemma \ref{lem:FMGC-feasibility})} \notag
\end{align}

\end{proof}
\section{Conclusion}
Part of the contributions of this work are conceptual: namely, the introduction of a multi-tier network design model and relating it to the multi-objective graph min-cut problem. 
Our $k$-tier Flexible Graph Connectivity problem is a multi-tier generalization of the  $(p, q)$-FGC problem and it captures non-uniform failure scenarios with nested tiers of edge vulnerability.
We gave three approximation results for fixed constant $k$, all of which are inspired by previous work for $(p, q)$-FGC with additional simplifications that help in dealing with the generalized model. 
$k$-tier FGC opens a new direction in network design by imposing hierarchical connectivity requirements. 
An interesting question is whether the dependence on $k$ and $n$ in the approximation factor of Theorem \ref{thm:log-approx} can be eliminated in exchange for a dependence on tier requirement parameters. If $k$ is part of the input, then feasibility verification is already likely to be co-NP-hard (based on connections to multi-objective graph min-cut which is strongly NP-hard when the number of objectives is part of input). 
For $p$-ECSS, there exists a constant approximation and for $(p, q)$-FGC, Ibrahimpur and V\'{e}gh \cite{IV25} designed an $O(\log{n})$-approximation.  For $k$-tier FGC, it is natural to ask whether an approximation that depends only on tier requirements is achievable for every fixed constant $k$. 

\bibliographystyle{abbrv}
\bibliography{references}
\end{document}